\documentclass{ecai}
\usepackage{graphicx}
\usepackage{latexsym}
\usepackage{balance} 
\usepackage{hyperref}
\usepackage{enumerate}
\usepackage[inline]{enumitem}
\usepackage{hhline,multirow}%
\usepackage{amsmath,amssymb,amsfonts}%
\usepackage{amsfonts}       
\usepackage{nicefrac}       

\usepackage{booktabs} 
\usepackage[caption = false]{subfig}
\usepackage{float}

\makeatletter
\DeclareRobustCommand{\qed}{%
  \ifmmode 
  \else \leavevmode\unskip\penalty9999 \hbox{}\nobreak\hfill
  \fi
  \quad\hbox{\qedsymbol}}
\newcommand{\openbox}{\leavevmode
  \hbox to.77778em{%
  \hfil\vrule
  \vbox to.675em{\hrule width.6em\vfil\hrule}%
  \vrule\hfil}}
\newcommand{\qedsymbol}{\openbox}
\newenvironment{proof}[1][\proofname]{\par
  \normalfont
  \topsep6\p@\@plus6\p@ \trivlist
  \item[\hskip\labelsep\itshape
    #1.]\ignorespaces
}{%
  \qed\endtrivlist
}
\newcommand{\proofname}{Proof}
\makeatother

\newtheorem{definition}{Definition}
\newtheorem{problem}{Problem}
\newtheorem{lemma}{Lemma}
\newtheorem{example}{Example}
\newtheorem{proposition}{Proposition}

\newtheorem{remark}{Remark}

\def \WWl{\mathcal{W}_{\lambda,C}}
\def \BB{\mathcal{B}}
\def \DD{\mathcal{D}}

\def \QQ{\mathcal{Q}}
\def \EE{\mathcal{E}}

\def \AA{\mathcal{A}}

\def\B{{\mathbb B}}

\def\argmin{{\rm argmin}}
\def\argmax{{\rm argmax}}

\begin{document}

\begin{frontmatter}

\title{A Bilevel Formalism for the Peer-Reviewing Problem}

\author[A]{\fnms{Gennaro}~\snm{Auricchio}\thanks{Corresponding Author. Email: ga647@bath.ac.uk}}
\author[B]{\fnms{Ruixiao}~\snm{Zhang}}
\author[A]{\fnms{Jie}~\snm{Zhang}}
\author[B]{\fnms{Xiaohao}~\snm{Cai}} 

\address[A]{Department of Computer Science, University of Bath}
\address[B]{School of Electronic and Computer Science, University of Southampton}

\begin{abstract}
Due to the large number of submissions that more and more conferences experience, finding an automatized way to well distribute the submitted papers among reviewers has become necessary. We model the peer-reviewing matching problem as a {\it bilevel programming (BP)} formulation. Our model consists of a lower-level problem describing the reviewers' perspective and an upper-level problem describing the editors'. Every reviewer is interested in minimizing their overall effort, while the editors are interested in finding an allocation that maximizes the quality of the reviews and follows the reviewers' preferences the most. To the best of our knowledge, the proposed model is the first one that formulates the peer-reviewing matching problem by considering two objective functions, one to describe the reviewers' viewpoint and the other to describe the editors' viewpoint. We demonstrate that both the upper-level and lower-level problems are feasible and that our BP model admits a solution under mild assumptions. After studying the properties of the solutions, we propose a heuristic to solve our model and compare its performance with the relevant state-of-the-art methods. Extensive numerical results show that our approach can find fairer solutions with competitive quality and less effort from the reviewers.\footnotemark
\end{abstract}

\end{frontmatter}

\stepcounter{footnote}
\footnotetext{Our code website: \url{https://github.com/Galaxy-ZRX/Bilevel-Review}.}

\section{Introduction}

%
The peer-review process is the procedure followed by scientific journals to establish the novelty and correctness of research papers submitted for publication \cite{shah2022challenges,jana2019history}. 
The principle behind this process is simple.
Once the editors of a journal receive a submission, they invite a sufficient number of reviewers to review it and, depending on the reviewers' reports, the editors decide whether the paper meets the criteria for publication or not.
To make this process meaningful, there must be a concordance between the topic of the paper and the expertise of the reviewers, ensuring both the authors of the paper and the editors that the reviewers are competent on the topics discussed in the paper.
Although the peer-review procedure has been proven to be reliable for scientific journals, it shows some flaws when applied to large conferences. 
In this case, all authors need to submit their papers by a specified deadline, after which all the submitted papers will be allocated to some reviewers by the editors.
Since most conferences have only one submission deadline per year, all the papers are submitted in a few days; consequently, the editors need to handle a large number of papers at the same time.
In particular, the editors have to determine the allocation hastily since this allocation process needs to respect the tight schedule of the conference.
Due to the large number of papers the editors usually receive and the lack of time, the allocation process becomes unaffordable through classic means.
To meet this need for time efficiency, several automatized ways to determine a peer-reviewer allocation have been proposed.
As we will detail in the related work below, the vast majority of the approaches assume that the best peer-reviewing matching is the minimum/maximum of a suitable objective function.
This allows rephrasing the whole problem as an optimization problem whose objective encodes an appealing property that the peer-reviewing allocation should have.
Albeit every model finds an allocation following its own criterion, all the models tacitly assume that the editors are the agents that actively decide the allocation.
This is limiting for two reasons.
One is that the efficiency of these methods is bound to how accurate the computation of the objective function is;
the other is that reviewers are often allowed to interfere in the assignment procedure by refusing to review some papers or by bidding for them. 
To represent the interplays between different agents (\textit{i.e.}, editors and reviewers), we propose and study a model that phrases the peer-reviewing allocation problem through a {\it bilevel programming (BP)} formulation.
Bilevel optimization has been proven to be a valuable expedient to describe several allocation problems in a multi-agent framework, such as the ones concerning resource allocation \cite{xu2013bilevel}, traffic engineering  \cite{9170527}, genomic problems \cite{burgard2003optknock}, and Knapsack Problems \cite{denegre2011interdiction}.
The appeal of BP problems lies in the fact that relating the solutions of different optimization problems well captures how the choice of one agent affects the others.
For a complete discussion on the BP problems, we refer the reader to \cite{dempe2002foundations}.

{\bf Related Work.} 
At its core, finding a peer-reviewing matching is a matching problem for indivisible goods.
Matching problems were first introduced in the first half of the twenty century by Hitchcock \cite{Hitchcock1941} and Kantorovich \cite{Kantorovitch1958}.
Ever since they were introduced, this class of problems found several different applications, such as stable marriage \cite{irving1994stable,mcvitie1971stable}, workers and job matching \cite{Easterfield1946,Thorndike1950}, resource allocation \cite{feng2013device,7295474}, and biomedical data analysis \cite{auricchio2018computing}.

It was observed that the first attempt at modeling the peer-reviewing process as a matching problem dates back to $2007$; and then an integer linear programming (ILP) model was proposed, whose solution is an allocation minimizing the possible complaints of the reviewers.
This was done with or without considering the reviewers' expertise \cite{goldsmith2007ai}. 
In \cite{garg2010assigning,lesca2010lp}, the authors improved this model by searching for a fair allocation, \emph{i.e.}, an allocation that equalizes every reviewers' payoff as much as possible. 
In \cite{garg2010assigning}, this is accomplished by maximizing the minimum of the reviewers' payoffs, while in \cite{lesca2010lp}, the authors try to implement classic social inequality functions, such as the Gini's social evaluation function.
Another similar approach is taken in \cite{DBLP:conf/aaai/DickersonSSX19}, where the authors proposed a new category of matching problems, called \emph{diverse matching problems}.
%
%
Despite the differences of those models, they all share the idea that the optimal peer-reviewing allocation is characterizable as the minimum/maximum of an objective function over a suitable restricted set of matching.
Moreover, all these routines are able to detect a solution without gathering data from the reviewers.
Indeed, all the objective functions used by these procedures are assumed to be known \textit{a priori} or to be computable through machine learning means using data from open-access sites, such as Google Scholar.
These concepts are also at the base of the Toronto Paper Matching System (TPMS) \cite{toronto} -- one of the most credited and used peer-reviewing matching procedures.
%
%

%


{\bf Paper Outline.}
In Section \ref{sec:basic_notions}, we briefly review the classic ILP model for the peer-reviewing matching problem.
In Section \ref{sect:bilevel}, we introduce our BP model for the peer-reviewing matching problem.
We describe and study the upper-level problem (ULp) and lower-level problem (LLp) of our formulation and give a set of conditions that ensure the existence of a solution.
%
%
In Section \ref{sec:num_exp}, a heuristic solution is introduced to our BP model and its performance compared with the ones obtained by the relevant state-of-the-art.
%
The results show that our model can find quasi-optimal solutions that are much fairer than the classic ones. 
Finally, we conclude in Section \ref{sec:conclusion} and outline some future research directions.
Due to space limit, we report the proofs of the proposed statements, the discussion on the Secondary Variational Problem, and additional numerical results in the Appendix.

\section{Maximum Edge-weighted Matching formulation}
\label{sec:basic_notions}

In this section, we recall the ILP model on which the TPMS is based \cite{toronto}.
Throughout the paper, we denote with $\mathbb{B}_{nm}$ the set of $n\times m$ binary matrices.
Given $A,B\in\mathbb{B}_{nm}$, we say that $A\le B$ if $a_{i,j}\le b_{i,j}$, $\forall i\in [n], \; j \in [m]$, where $[\ell]$ denotes the set containing the first $\ell$ natural numbers, \emph{i.e.}, $[\ell]:=\{1,2,3,\dots,\ell\}$.

In mathematical terms, peer-reviewing matching can be described as the solution of a maximum edge-weighted matching problem over a bipartite graph $G$.
The two sides of the bipartite graph $G$ are composed of the set of papers $\mathfrak{P}=[n]$ and  the set of reviewers $\mathfrak{R}=[m]$.
The bipartite graph is therefore given by $G=([n]\cup [m], [n]\times [m])$ and any matrix $X\in\mathbb{B}_{nm}$ describes a possible papers allocation.
Following the classic conventions, the edge $(i,j)$ is active, \emph{i.e.}, $X_{i,j}=1$, if and only if paper $i$ is allocated to reviewer $j$.

Since the editors have to ensure a certain amount of reviews for every paper and no reviewer can be overburdened with papers to review, every feasible matching has to comply with some restriction that bounds the number of active edges connected to every vertex of the bipartite graph.
The linear objective function of the problem is determined by the edge weight matrix $W=\{w_{i,j}\}_{(i,j)\in G}$.
The weight $w_{i,j}$ of the edge $(i,j)$ represents the expected quality of the review that paper $i$ receives from reviewer $j$, hence we refer to the matrix $W$ as the \emph{quality} matrix.
In this framework, the optimal peer-reviewing matching is characterized as the solution of the maximum edge-weighted matching problem induced by the quality matrix $W$ over the bipartite graph $G$, since it describes the matching maximizing the overall quality of the reviews.
We below detail the constraints the editors must comply with and how the matrix $W$ is determined.

\paragraph{Problem Constraints.}

We have two sets of constraints, one for each side of the bipartite graph.
Their role is to regulate the number of reviews that every paper receives and to limit the number of reviews that every reviewer can be asked to perform.
In detail, the constraints are as follows.
\begin{enumerate*}[label=(\roman*)]
\item The set of constraints that is imposed over the set of reviewers is $\sum_{i\in [n]} X_{i,j}\le U_j$, $\forall j \in [m]$. 
    To do so, we restrict the number of papers the editors can assign to each reviewer.
    The value $U_j$ is therefore the maximum number of papers that can be allocated to reviewer $j$.
\item The set of constraints that is imposed over the set of papers is $l_i \le \sum_{j\in [m]}X_{i,j}\le u_i$, $\forall i \in [n]$.
    To do so, we bound the editors to assign every paper $i$ to at least $l_i$ and at most $u_i$ reviewers.
    \item We assume every paper is indivisible, imposed by the restriction $X\in\mathbb{B}_{nm}$.
\end{enumerate*}
Finally, the optimal peer-reviewing assignment is retrieved by solving the following ILP problem.
%

%
\begin{problem}
\label{pr:classic}
Given $n,m\in\mathbb{N}$, and a quality matrix $W$, the classic ILP model for the peer-reviewing assignment is 
\begin{align*}
   \max_{X\in\mathbb{B}_{nm}}\langle W,X \rangle,\quad\text{s.t.} \quad l_i \le \sum_{j\in [m]}X_{i,j} \le u_i, \quad \,\sum_{i\in [n]} X_{i,j}\le U_j,
\end{align*}
where the constraints hold for every $i\in [n]$ and $j\in [m]$ and $\langle \cdot , \cdot \rangle$ is the scalar product in the Euclidean space.
\end{problem}

It is easy to see that, as long as $l_i\le u_i$, $\forall i\in [n]$ and $\sum_{i \in [n]}l_i\le \sum_{j\in [m]}U_j$, Problem \ref{pr:classic} is well-defined and admits a solution \cite{edmonds1965maximum}.
%
%

\paragraph{Quality Matrix $W$ and Its Role.}
\label{sec:workload_matrix}
Determining a meaningful quality matrix $W$ is of key importance for every allocation method based on the ILP model described in Problem \ref{pr:classic}.
%
%
For this reason, retrieving a quality matrix $W$ that reliably represents the expertise of the reviewers is still a topic of study and discussion.

One way to retrieve the matrix $W$ is to associate every paper and reviewer to a set of labels that summarizes the contents and the fields of expertise, respectively.
Since both the contents of the paper and the fields of expertise are drawn from the same set of topics,  it is possible to represent these labels as vectors over an $r$ dimensional space, where $r$ is the total number of possible topics.
If we denote with $v_{p_i}$ the vector describing the topics of the paper $i$ and with $v_{r_j}$ the expertise vector of reviewer $j$, we can then compute the quality of the review performed by reviewer $j$ on paper $i$ as the scalar product between $v_{p_i}$ and $v_{r_j}$, \cite{sugiyama2010scholarly}.
This metric does make sense since the more the vectors are aligned, the higher will be their scalar product.
Moreover, if all the vectors have only positive entries, all the scalar products will return a positive value.
Since there is no canonical way to represent papers and reviewers in an $r$ dimensional space, this method highly depends on what embedding we use to translate papers and reviewers into vectors.
Usually, the embedding is determined through neural network structures that use data reported from the reviewer or available from sites on the internet (such as Google Scholar \cite{toronto}).
Despite the remarkable results obtained by those procedures, it is a common belief that evaluating the quality of a review only from the publication records of the reviewer is limiting \cite{toronto}.
Indeed, these models do not consider other factors that are instead related to the effort the reviewer would put in reviewing the paper, such as personal interest, personal conflict, or even time at the disposal.
%
%


\section{Bilevel Formulation of the Assignment Problem}
\label{sect:bilevel}

In this section, we first introduce our BP formulation for the peer-reviewing matching problem, followed by studying the LLp and ULp and showing that our model admits a solution under mild conditions.
%
%
We then highlight the relationship between our model and the classic literature.
Throughout the paper, we assume that the editor in charge is just one person.

\subsection{Bilevel Programming Model}

Given a set of papers $\mathfrak{P}=[n]$ and a set of reviewers $\mathfrak{R}=[m]$, we define our BP peer-reviewing matching problem as follows.

\begin{problem}
\label{problem:bilevel}
Given $l_i$, $u_i$ and $U_j$ as in Section \ref{sec:basic_notions}, we consider the following problem:
\begin{eqnarray}
\label{eq:bivel_prob}
\arraycolsep=2pt\def\arraystretch{1.25}
\begin{array}{ll} 
\underset{Z,X\in \mathbb{B}_{nm}}\max &\;\;\; \langle W_E,X\rangle + \langle Y^*,X \rangle\\
\quad \quad {\rm s.t.} &
 \quad l_i \le\sum_{j\in [m]}X_{i,j}\le u_i, \quad \sum_{i\in [n]}X_{i,j}\le U_j,\\
& \quad\sum_{i\in [n]}Z_{i,j}=U_j+\phi_j,\quad X \le E-Z+Y^*, \\
    & \quad Y_j^* \in\underset{ Y_j\in \mathbb{B}_{n}}\argmin \; \langle (W_R)_j,Y_j\rangle, \\
    &\quad \sum_{i\in [n]}Y_{i,j}=U_j, \quad \text{and}\quad 0\le Y_j\le Z_j,
\end{array} 
\end{eqnarray}
where $E$ is a $n \times m$ matrix whose entries are all equal to $1$.
We denote with $W_E$ the $n\times m$ matrix  describing the qualities of all the possible reviews from the editor's point of view, with $W_R$ the $n\times m$ matrix describing the reviewers' efforts, and with $\phi_j\in \mathbb{N}$ the number of papers that the reviewers can refuse to review.
%
Finally, since we assume that every reviewer bids independently, the LLp is a component-wise minimization, \textit{i.e.}, $Y^*$ is the matrix whose columns minimize the function $Y_{:,j}\to\langle (W_R)_{:,j},Y_{:,j} \rangle$.
%
%
%
\end{problem}

In Table \ref{tab:variables_and_parameter_2}, we report the variables of our model and their meaning.

\begin{remark}
It is well-known that every BP problem describes a Stackelberg game \cite{von1952theory,luo1996mathematical}, and this is no exception.
In our case, the leader is the editor who proposes $Z$ to the followers, which are the reviewers.
Afterwards, the reviewers report a vector $Y$ according to their preferences.
Every couple $(Z,Y)$ uniquely defines a maximum edge-weighted matching problem, whose objective value corresponds to the payoff of the editor\footnote{If the couple $(Z,Y)$ leads to an unfeasible problem, the payoff is $-\infty$.
The reviewers' payoff is the total effort of the papers they bid for.}.
Notice that the BP formulation we introduced is also equivalent to the following three-phase procedure.
First, the editor proposes a set of papers, described by the vector $Z_j$, to every reviewer $j$.
Second, the reviewers select a set $Y_j$ out of the set $Z_j$. 
The vector $Y_j$ represents the papers that reviewer $j$ would review if asked.
Finally, the editor gathers the replies of all the reviewers and searches for a papers-reviewers allocation that maximizes the quality and matches the reports given by the reviewers as much as possible.
For this reason, given any feasible triplet to problem \eqref{eq:bivel_prob}, namely $(X,Y,Z)$, we refer to $Z$ as the editor's proposal, to $Y$ as the reviewers' bidding, and to $X$ as the final assignment.
A graphical description of this procedure is also given in Figure \ref{fig:figures_bis}.
\end{remark}

\begin{figure*}[ht!]

\begin{center}
    \subfloat[The initial complete bipartite graph.]{\label{fig:2_1}\includegraphics[width = 0.16\textwidth]{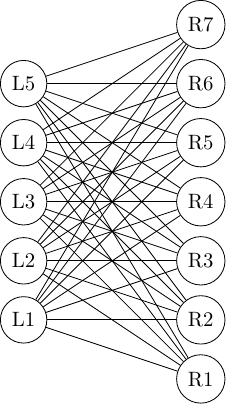}} 
\hfil
\subfloat[The graph describing the proposal $Z$ of the editor.]{\label{fig:2_2}\includegraphics[width = 0.16\textwidth]{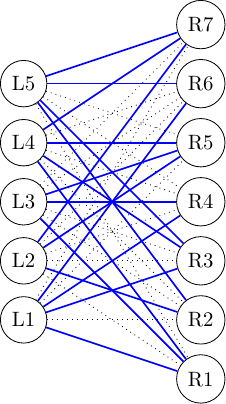}}
\hfil
\subfloat[The graph describing the bidding of the reviewers $Y$.]{\label{fig:2_3}\includegraphics[width = 0.16\textwidth]{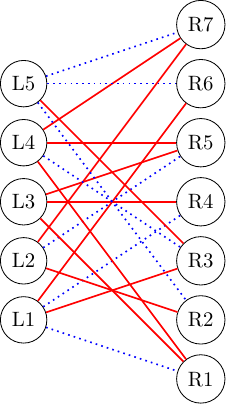}}
\hfil
\subfloat[The graph on which the editor solves the final matching problem.]{\label{fig:2_4}\includegraphics[width = 0.16\textwidth]{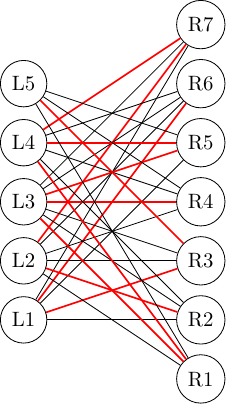}} 
\caption{An example of the allocation procedure described by our BP problem. 
In (\ref{fig:2_1}), we have the complete bipartite graph. The left side contains the papers, while the right side contains the reviewers. 
Reviewers $R6$, $R4$ and $R2$ have to review at most $1$ paper, while the other reviewers have to review at most $2$. 
All reviewers have a degree of freedom $\phi_j$ equal to $1$.
In (\ref{fig:2_2}), we report in blue the editor's proposal $Z$ done by the editor, while the dotted lines represent the edges that have not been proposed by the editor. 
In (\ref{fig:2_3}), we report in red the edges the reviewers' bidding $Y$, while the blue dotted lines represent the reviews that the reviewers have declined.
Finally, in (\ref{fig:2_4}), we report the final graph on which the editor has to solve a maximum edge-weighted matching problem.  \label{fig:figures_bis}
}
\end{center}

\end{figure*}

As in the previous model, we assume that the editor aims to maximize the quality of the peer-reviewing matching.
In our case, however, both the objective function and the constraints of the problem depend on the reviewers' biddings, \textit{i.e.}, the solution of the LLp.
The influence that the LLp has on both the objective function and the set of feasible matching of the ULp is the way in which we model the interaction between the reviewers and the editor.
\begin{table}[t]
    
    \begin{center}
    \caption{List of variables and parameters used in our BP model.}
    \label{tab:variables_and_parameter_2}
        \begin{tabular}{ p{0.35cm}|p{7.45cm}}
     \midrule
     $l_i$ & minimal number of reviews needed by paper $i$\\
     $u_i$ & maximal number of reviews needed by paper $i$\\
     $U_j$ & number of papers that reviewer $j$ bids for and maximum number of papers that $j$ may review\\
     $Z_j$ & vector describing the set of papers proposed to reviewer $j$\\
     $Y_j$ & vector describing the set of papers that reviewer $j$ bids for\\
     $Y^*_j$ & vector that minimizes reviewer $j$'s effort\\
     $X$ & final allocation decided by the editor\\
     $\phi_j$ & degree of freedom of reviewer $j$ \\
     $W_E$ & quality matrix from the editor's point of view\\
     $W_R$ & reviewers' effort matrix\\
     
     \midrule
    \end{tabular}
    \end{center}
     
\end{table}

\begin{remark}[The effort matrix $W_R$ and its role]

To describe the reviewers' point of view, we introduce a matrix $W_R$ that assigns an effort value to every reviewer-paper couple.
%
%
%
When the number of papers is small, it is possible to retrieve $W_R$ by directly asking the reviewers.
As an alternative, we can build a proxy version of $W_R$ using machine learning.
%
Indeed, there are machine learning methods able to find a complete preference order over a set of papers by only having a partial declaration of preferences \cite{toronto}.
In previous models these values were used to enhance the quality matrix $W_E$ by refining the values that are less accurate due to a lack of information about the reviewer.
Instead, in our model, we use the preference order to describe the criterion used by the reviewer during the bidding phase.

\end{remark}


\subsection{Lower-Level Problem}
\label{subsec:llp}

We now consider the LLp, which captures the reviewers' point of view.
We first describe the objective function and the constraints, and then we show that the LLp has a solution and study its properties.

\subsubsection{Reviewer's Objective Function}

Once the reviewers receive the editor's proposal $Z$, they have to bid for part of the papers proposed and report their bidding to the editor.
In the following, we denote with $Y_{j}$ the $j$-th column of the matrix $Y$, so that $Y_j=Y_{:,j}$.
Similarly, $(W_R)_j$ and $Z_j$ denote the $j$-th column of the matrices $W_R$ and $Z$, respectively.
We model the interests of the reviewers through a linear objective function that evaluates the effort required by every reviewer $j$ to review a set of papers $Y_j$.
In this framework, the bidding criterion adopted by reviewer $j$ is captured by the minimization of the objective function $Y_j\to \langle  (W_R)_j, Y_j\rangle$.
Every entry of the reviewers' effort vector $((W_R)_j)_i$ describes how much effort reviewer $j$ would spend to review paper $i$. 
To make the model meaningful, we assume that each entry of the vector $(W_R)_j$ is positive, which means that reviewing is never effortless.

\subsubsection{Reviewer's Constraints}

Every reviewer $j$ has to comply with two sets of constraints.
Since reviewers are anonymous to each other, we assume that the constraints imposed on one reviewer do not depend on the constraints imposed on other reviewers. 
Thus, we describe the constraints only for a fixed reviewer, namely $j$. 
To impose the indivisibility of the reviews, every $Y_j$ is assumed to have binary entries.
Finally, the set of constraints with which the reviewers have to comply is as follows.
\begin{enumerate*}[label=(\roman*)]
\item The \emph{consistency constraint}, that is $Y_j\le Z_j$, where $Z_j$ is the proposal advanced by the editor; this constraint imposes that every reviewer cannot bid for a paper that does not belong to the pool that the editor proposed.
\item The \emph{quantity constraint}, that is $\sum_{i\in [n]}Y_{i,j}= U_j$; this constraint forces every reviewer to bid for $U_j$ papers.
Since the review effort is always non-negative, imposing the quantity constraint is equivalent to imposing the following constraint $\sum_{i\in [n]}Y_{i,j} \ge U_j$.
\end{enumerate*}

\paragraph{Uniqueness of the Solution and the Equivalence Result.}
To conclude, we study the solution of the LLp.
As we will see, the set of conditions needed to ensure the uniqueness of the solution and the equivalence with the relaxed problem is natural to assume.

\begin{lemma}
\label{property} 
Given two integers $0<a<b$, and let $\{w_{k}\}_{k=1,\dots,b}$ be a set of positive and increasingly ordered real values, \textit{i.e.}, $0\leq w_1\leq  \cdots \leq w_a\leq w_{a+1}\leq\cdots\leq w_{b}$. 
If $x_i\in [0,1]$, $\forall i\in[b]$ and $\sum_{i=1}^b x_i=a$, then it holds $\sum_{i=1}^a w_i \leq \sum_{i=1}^b w_i x_i$.
\end{lemma}

We now prove that the ILP problem in the LLp is equivalent to the LP problem obtained by relaxing the constraint $Y_j\in\{0,1\}$.
%
%

\begin{theorem}
\label{lemma-equivelance}
For any given $Z\in\B_{mn}$, the LP problem 
\begin{eqnarray}
\label{lower-level-problem-primal_0}
Y^* \in\underset{0\le Y\le Z}\argmin ~\langle W,Y \rangle,\quad {\rm s.t.} \ \ \sum_{i\in [n]}Y_{i,j}=U_j, 
\end{eqnarray}
and the ILP problem
\begin{equation} 
\label{lower-level-problem-primal}
\widehat Y \in\underset{Y \in \B_{mn}, Y\le Z }{\rm argmin}~\langle  W, Y\rangle,\quad {\rm s.t.} \ \ \sum_{i\in [n]}Y_{i,j}=U_j,
\end{equation}
attain the same optimal value, \textit{i.e.}, $\langle  W, Y^*\rangle=\langle  W, \widehat Y\rangle$. 
Moreover, if $\forall j\in[m]$, the set $\{ W_{i,j}: (i,j) ~\text{s.t.}~ Z_{ij}=1\}$ does not contain two equal values, then the solution is unique and satisfies  $\widehat{Y} \in \B_{mn}$, \textit{i.e.}, Problems \eqref{lower-level-problem-primal_0} and \eqref{lower-level-problem-primal} are equivalent.
\end{theorem}

Notice that the hypothesis of Theorem \ref{lower-level-problem-primal} is satisfied whenever every reviewer has a strict preference order over the set of papers.
Finally, we show that minimizing the component-wise reviewers' effort is the same as minimizing the global reviewers' effort.

\begin{theorem}
    \label{prop:prob_simp}
For every given $Z$, $Y$ is a solution of the LLp if and only if $Y$ minimizes the function $Y\to \sum_{i\in [n]}\sum_{j\in [m]} (W_R)_{i,j}Y_{i,j}$ under the constraints $\sum_{i\in [n]}Y_{i,j}-U_j\ge 0$ and $Z_{i,j}-Y_{i,j}\ge 0$.
\end{theorem}

%

\subsection{Upper-Level Problem}

We now consider the ULp, which is the problem that captures the editor's point of view.
We first describe the role of both the objective function and the constraints, and
then we show that the solution to the ULp problem always exists under mild assumptions.

\subsubsection{The Editor's Objective Function}

In our model, the editor needs to find an allocation that maximizes a linear objective function, described by the matrix $W_E$ and the scalar product of the variables $X$ and $Y$.
While $\langle W_E,X \rangle$ represents the quality of matching $X$, the scalar product between $X$ and $Y$ describes how much the final assignment $X$ meets the reviewers' bidding $Y$.
%
%
We, therefore, define the following quantity.

\begin{definition}
\label{def:perfect}
Given a solution $(X,Y,Z)$ to the BP Problem \ref{problem:bilevel}, we define its Accordance Percentage as $AC(X,Y,Z):=\frac{\langle X , Y \rangle}{\langle X , X \rangle}$.
\end{definition}

By definition, we have $AC(X,Y,Z)\in [0,1]$ and $AC(X,Y,Z)=1$ if and only if $X\le Y$, \textit{i.e.} if and only if every reviewer receives only papers they bid for.
%
If $AC(X,Y,Z)=1$ holds, we say that the solution $(X,Y,Z)$ is \emph{perfect}.

\subsubsection{The editor's constraints}

In our model, the editor controls two variables, $Z$ and $X$. 
%
%
Both $Z$ and $X$ are $n\times m$ binary matrices that are subject to two different sets of constraints below.
\paragraph{Constraints over $Z$.} We require the proposal $Z$ to comply with only one set of constraints, which we call \emph{freedom constraint}.
The role of these constraints is to determine the number of papers the editor has to propose to every reviewer.
In mathematical terms, this means that $\sum_{i\in [n]}Z_{i,j}=U_j+\phi_j$, $\forall j\in [m]$, where the constants $\phi_j\ge 0$ are the degrees of freedom.
In fact, since the reviewer $j$ has to bid for $U_j$ papers, the parameter $\phi_j$ corresponds to the number of papers that the reviewer can refuse to review.

\paragraph{Constraints over $X$.} 
The final assignment $X$ has to comply with the  constraints as follows.
\begin{enumerate*}[label=(\roman*)]
\item The \emph{feasibility constraints}, that is, $l_i\le\sum_{j\in [m]}X_{i,j}\le u_i$, $\forall i \in [n]$ and $\sum_{i\in [n]} X_{i,j}\le U_j$, $\forall j \in [m]$.
    These constraints bound the editor to assign every paper $i$ to a number of reviewers between $l_i$ and $u_i$ and to assign at most $U_j$ reviews to every reviewer $j$.
    \item The \emph{consistency constraint}, that is $X\le E-Z+ Y^*$, where $Y^*$ is the solution to the LLp and $Z$ is the editor's proposal. 
    This constraint bounds the editor to assign a paper to a reviewer only if the reviewer has not refused it before.
\end{enumerate*}

\paragraph{The Solutions of the ULp}

In the previous section, we proved that the LLp admits a solution, regardless of what the editor proposes.
As the next example shows, the same does not hold for the ULp.

\begin{example}
\label{ex:fairness}
Let us consider the case in which we have three papers, namely $p_1$, $p_2$, and $p_3$, and two reviewers, namely $r_1$ and $r_2$.
We assume every paper needs exactly one review (\textit{i.e.}, $l_i=u_i=1$) and that every reviewer can be asked to review at most $2$ papers (\textit{i.e.}, $U_j=2$).
Finally, we assume that $W_R$ is defined as $W_R=\begin{pmatrix} 11 & 10 & 1\\3 & 2 & 1\end{pmatrix}$.
Both reviewers have the same preference order over the set of papers, \textit{i.e.}, $p_1>p_2>p_3$.
It is then easy to see that, if $\phi_1,\phi_2>0$, no reviewer will give its consensus to review paper $p_3$, making the ULP, and hence the whole BP problem, unfeasible.
\end{example}

Example \ref{ex:fairness} points out that allowing the reviewers to review or refuse too many papers leads to an unsolvable problem in some pathological cases.
In fact, if we ask the second reviewer to review only one paper, that is $U_2=1$, the BP problem described in Example \ref{ex:fairness} becomes feasible.
In general, the editor can ensure the existence of a solution to Problem \ref{problem:bilevel} by tuning the parameters of the problem.

\begin{theorem}
\label{prop:feasibility}
If, in addition to the assumptions that make Problem \ref{problem:bilevel} feasible, it holds $\max_{j\in [m]}\;\phi_j + 2\max_{j\in [m]} \; U_j\le n$, then there exists a feasible triplet to Problem \ref{eq:bivel_prob}.
In particular, under these assumptions, Problem \ref{eq:bivel_prob} has a solution.
\end{theorem}


%
Notice that whenever the number of papers $n$ is large enough, the hypothesis of Theorem \ref{prop:feasibility} is likely to be satisfied, since both $U_j$ and $\phi_j$ are parameters bounded to the human capacity.
For example, if $n=30$, the conditions will hold even if we require every reviewer to review at most $8$ papers while allowing everyone to reject up to $10$ papers.
Finally, even when the number of papers is small, the editor can decrease $\phi_j$ and $U_j$ in order to make them suitable.

\subsection{Fairness and the relation with the ILP model}

In Example \ref{ex:fairness}, we showed that when the reviewers share the same preference order over the set of papers, finding a solution to the BP formulation might be impossible. 
In the following, we show that this phenomenon is also due to an intrinsic fairness property that perfect solutions to the BP problem possess.

\begin{definition}
Let $X$ be a matching over the bipartite graph $G$ and let $W_R$ be the reviewers' effort matrix.
Given a vector $\Phi:=(\phi_1,\dots,\phi_m)\in\mathbb{N}^{m}$, we say that $X$ is $\Phi$-weakly fair if, $\forall j\in [m]$, the set $BC_j(X):=\big\{i\in [n], \ \text{s.t.} \ X_{i,j}=0 \ \& \ (W_R)_{i,j} \ge \min_{ X_{i,j}=1}(W_R)_{i,j})\big\}$, contains at least $\phi_j$ elements.
\end{definition}

Notice that the set $BC_j(X)$ contains all the papers in $[n]$ that require more effort than the highest effort-requiring paper assigned to $j$ according to $X$. 
Therefore, a solution is $\Phi$-weakly fair if every reviewer $j$ is not allocated with its $\phi_j$ worst picks.
Using the notion of a perfect solution, we are able to relate $\Phi$-weakly fair solutions of the classic ILP model to the perfect solutions of the BP model.

\begin{proposition}
\label{prop:dots}
There exists a perfect solution $(X,Y,Z)$ that maximizes the quality if and only if the ILP problem admits a $\Phi$-weakly fair solution.
\end{proposition}

To conclude, we show that our model extends the classic ILP formulation described in Problem \ref{pr:classic}.
In particular, whenever we set $\phi_j=0$, $\forall j\in [m]$, the solution of the LLp is actually determined by the editor.
This allows us to find a bijection between the solutions of the classic ILP model and the solutions of our BP model.

\begin{proposition}
\label{prop:dict}
Set $\phi_j=0$, $\forall j\in [m]$.
Then, given a solution $(X,Y,Z)$ of the BP problem, we have that $X$ solves Problem \ref{pr:classic}.
Vice-versa, if $X$ is a solution to Problem \ref{pr:classic}, then all the triplets $(X,Y,Z)$ such that $X\subset Y=Z$ are perfect solutions to the BP problem. 
\end{proposition}

In Appendix B, we delve further into the relationships between our BP model and previous models. 
In particular, we focus on the class of diverse matching proposed in \cite{DBLP:conf/aaai/DickersonSSX19}, show how this specific model is related to a Secondary Variational Problem induced by Problem \ref{pr:classic}, and adapt it to our BP model.
%

\section{Numerical Experiments}
\label{sec:num_exp}

In this section, we report the results of our numerical experiments.
We first introduce the heuristic we use to solve the model.
We then describe the experimental setting, the implementation details, and introduce the metrics used to compare the outcomes of our experiments.
Finally, we comment on the results we find.
%

\subsection{Greedy Heuristic}

%
Due to the complexity of the BP model described in Problem \ref{problem:bilevel}, approaching the problem through Bilevel solvers is prohibitive.
%
%
For this reason, we propose a greedy heuristic that finds a solution by fixing the editor's proposal $Z$.
Indeed, once we fix $Z$, retrieving the reviewers' best response is simple since it is the minimum of $m$ independent LP problems.
Once we retrieve the best reply $Y_Z$, we obtain the final allocation by solving the maximum edge-weighted matching problem induced by $Z$ and $Y$.
For every $Z$ and $Y$, we denote the related final allocation with $X_{Z,Y}$.
Given a quality matrix $W_E$, we define the greedy proposal $Z_g$ as the solution to the following ILP problem
\begin{align*}
    Z_g = & \ \underset{Z\in \mathbf{B}_{nm}}{\argmax}\;\langle W_E,Z  \rangle,\quad\sum_{i\in [n]}Z_{i,j}=U_j+\phi_j,\quad\forall j \in [m].
\end{align*}
In particular, the greedy proposal $Z_g$ is the editor's proposal that maximizes the total quality of the papers proposed.
Following the previous notation, we define the greedy heuristic solution to Problem \ref{problem:bilevel} as $(X_g,Y_g,Z_g):=(X_{Z_g,Y_{Z_g}},Y_{Z_g},Z_g)$.
We notice that to build our heuristic solution, we need to solve at most three ILP problems. 
Therefore, complexity-wise, our approach is as costly as any method based on the resolution of an ILP problem. Moreover,
when the set of reviewers is large and heterogeneous enough, the greedy heuristic always finds a feasible peer review matching.

\begin{theorem}
    \label{thm:feas_cond}
    The greedy heuristic always finds a feasible peer review matching if, for every paper $i \in [n]$, there are at least $L=\sum_{i\in [n]}l_i$ reviewers that do not place $i$ among its top $K:=(\max_{j\in [m]}\phi_j+\max_{j\in [m]}U_j)$ picks.
\end{theorem}

Finally, we show that if the accordance percentage of the heuristic solution is equal to $1$, then the quality of the matching $X_g$ lower bounds the quality of the optimal solution to the BP problem.

\begin{theorem}
\label{them:heur_estimate}
Let $(X_g,Y_g,Z_g)$ be the triplet found by the greedy heuristic.
If $AC(X_g,Y_g,Z_g)=1$, then, for every optimal solution $(X^*,Y^*,Z^*)$ to the BP problem, we have  $\langle W_E, X_g\rangle\le \langle W_E, X^*\rangle$.
\end{theorem}

\subsection{Experimental Framework}
\label{subsect:exp}
We now detail the dataset, the parameters, and the solver we used during the implementation of our experiments.

{\bf Dataset.} We use the multi-aspect review assignment evaluation dataset \cite{karimzadehgan2008multi,karimzadehgan2009constrained}, which is a classic benchmark dataset from UIUC.
The dataset contains the information of $73$ papers accepted by SIGIR $2007$, as well as $189$ prospective reviewers who had published in the main information retrieval conferences. 
Moreover, the dataset provides $25$ major topics and, for each paper in the set, it provides a $25$-dimensional label on that paper based on a set of defined topics. 
Similarly, for each of the $189$ reviewers, a $25$-dimensional expertise representation is provided.
%
%
To the best of our knowledge, there is no available dataset from where we could retrieve the effort matrices needed for our tests.
For this reason, we rely on two synthetically generated datasets as follows. 
\begin{enumerate*}[label=(\roman*)]
\item The \textit{Aligned dataset}, in which there is a linear dependence between $W_R$ and $W_E$. We assume that $(W_R)_{i,j}=K-\big((W_E)_{i,j}+\chi)$, where $K$ is a constant ensuring $(W_R)_{i,j}>0$ and $\chi\sim \mathcal{N}(0,\sigma)$. 
In this case, maximizing the quality of the assignment and minimizing the total reviewers' effort have the same objective function up to a Gaussian noise $\chi$, which describes the personal preferences the reviewers may have over papers from the same area.
\item The \textit{Random dataset}, in which every entry of $W_R$ is defined as $(W_R)_{i,j}=\chi$, where $\chi$ is a random variable. 
In our test, we consider two possible laws for $\chi$: the Uniform distribution $\mathcal{U}[0,1]$ and the Exponential distribution $Exp(0.5)$.
%
%
In this scenario, the effort needed by any reviewer does not depend on the paper, but rather on external factors not related to the expected quality of the review, such as time at disposal and conflict of interest.
%
%
\end{enumerate*}

\begin{table*}[ht!]
\begin{center}

\caption[Experiments]{Comparison between $X_{BP}$ and $X_{ILP}$. Quantitative results for different values of $\phi$ and differently generated effort matrices. 
Every column represents a different framework and is characterized by the effort matrix, the random variable used to generate the matrices, and the degree of freedom.
For each framework, we report the averages over $250$ instances of the Quality Percentage, the Reviewers' Average Effort Ratio, the Fairness Ratio, and the Accordance Percentage of the heuristic solution.
%
}

\begin{tabular}{  c | @{}c@{}|@{}c@{} | @{}c@{}|@{}c@{} | @{}c@{}|@{}c@{}|@{}c@{}|@{}c@{}|@{}c@{}|@{}c@{} |@{}c@{} | @{}c@{} }
  \hline
  & \multicolumn{6}{c |}{Aligned } & \multicolumn{6}{ c }{Random } \\
  \cline{2-13}
  U=8 & 
  \multicolumn{3}{c|}{$\sigma=0.1$} & \multicolumn{3}{c |}{$\sigma=0.3$} & \multicolumn{3}{c|}{$X\sim \mathcal{U}[0,1]$}  & \multicolumn{3}{c }{$X\sim Exp(0.5)$} \\
  \cline{2-13}
  & $\;\;\phi=4\;\;$ & 
  $\;\;\phi=6\;\;$ & $\;\;\phi=8\;\;$ & 
  $\;\;\phi=4\;\;$ & $\;\;\phi=6\;\;$ & 
  $\;\;\phi=8\;\;$ & $\;\;\phi=4\;\;$ & 
  $\;\;\phi=6\;\;$ & $\;\;\phi=8\;\;$ & 
  $\;\;\phi=4\;\;$ & $\;\;\phi=6\;\;$ & 
  $\;\;\phi=8\;\;$ \\
  \hline
  $QP(X_{BP})$ & 0.99 & 0.99 & 0.99 & 0.98 & 0.98 & 0.97 & 0.96 & 0.95 & 0.93 & 0.97 & 0.96 & 0.94\\
  $RAER(X_{BP},X_{ILP})$ & 0.89 & 0.89 & 0.86 & 0.87 & 0.87 & 0.80 & 0.63 & 0.53 & 0.47 &  0.45 & 0.35 & 0.30   \\
  $FR(X_{BP},X_{ILP})$ & 0.65 & 0.75 & 0.57 & 0.61 & 0.61 & 0.48 & 0.35 & 0.24 & 0.19  & 0.16 & 0.09 & 0.07   \\
  $AC$ & 0.99 & 0.99 & 0.99 & 0.99 & 0.99 & 0.99 & 0.99 & 0.99 & 0.99 & 0.99 & 0.99 & 0.99 \\
  \hline
 
\end{tabular}
\label{tab:downsample_tr}
\end{center}
\end{table*}

\begin{table*}[ht!]
\begin{center}
\caption[Experiments]{Comparison between $X_{BP}$ and $X^{(t)}_{ILP}$. Quantitative results for different values of $\phi$ and differently generated effort matrices. 
Every column represents a different framework and is characterized by the effort matrix, the random variable used to generate the matrices, and the degree of freedom.
For each framework, we report the averages over $250$ instances of the Quality Percentage, the Reviewers' Average Effort Ratio, and the Fairness Ratio.
%
}


\begin{tabular}{  @{}c | c | @{}c@{}|@{}c@{} | @{}c@{}|@{}c@{} | @{}c@{}|@{}c@{}|@{}c@{}|@{}c@{}|@{}c@{}|@{}c@{} |@{}c@{} | @{}c@{} }
  \hline
   \multicolumn{2}{c |}{ } &  \multicolumn{6}{c |}{Aligned } & \multicolumn{6}{ c }{Random } \\
  \cline{3-14}
  \multicolumn{2}{c |}{ U=8 } & 
  \multicolumn{3}{c|}{$\sigma=0.1$} & \multicolumn{3}{c |}{$\sigma=0.3$} & \multicolumn{3}{c|}{$X\sim \mathcal{U}[0,1]$}  & \multicolumn{3}{c }{$X\sim Exp(0.5)$} \\
  \cline{3-14}
  \multicolumn{2}{c |}{} & $\;\;\phi=4\;\;$ & 
  $\;\;\phi=6\;\;$ & $\;\;\phi=8\;\;$ & 
  $\;\;\phi=4\;\;$ & $\;\;\phi=6\;\;$ & 
  $\;\;\phi=8\;\;$ & $\;\;\phi=4\;\;$ & 
  $\;\;\phi=6\;\;$ & $\;\;\phi=8\;\;$ & 
  $\;\;\phi=4\;\;$ & $\;\;\phi=6\;\;$ & 
  $\;\;\phi=8\;\;$ \\
  \hline
  \parbox[t]{2mm}{\multirow{3}{*}{\rotatebox[origin=c]{90}{$t=0.05$}}}& $QP(X_{ILP}^{(t)})$ & 0.99 & 0.99 & 0.99 & 0.99 & 0.99 & 0.99 & 0.99 & 0.99 & 0.99 & 0.99 & 0.99 & 0.99\\
  & $RAER(X_{BP},X_{ILP}^{(t)})$ & 0.96 & 0.95 & 0.95 & 0.93 & 0.91 & 0.89 &  0.72 & 0.60 & 0.53 & 0.59 & 0.47 & 0.39 \\
  & $FR(X_{BP},X_{ILP}^{(t)})$ & 0.77 & 0.77 & 0.76 & 0.70 & 0.65 & 0.62  & 0.43 & 0.29  & 0.23  & 0.26 & 0.16 & 0.11 \\
  \hline
  \parbox[t]{2mm}{\multirow{3}{*}{\rotatebox[origin=c]{90}{$t=0.1$}}}& $QP(X_{ILP}^{(t)})$ & 0.99 & 0.99 & 0.99 & 0.99 & 0.99 & 0.99 & 0.99 & 0.99 & 0.99 & 0.99 & 0.98 & 0.98\\
  & $RAER(X_{BP},X_{ILP}^{(t)})$ & 0.95 & 0.95 & 0.95 & 0.95 & 0.93 & 0.91 &  0.78 & 0.66  & 0.58 & 0.71 & 0.57 & 0.47 \\
  & $FR(X_{BP},X_{ILP}^{(t)})$ & 0.78 & 0.77 & 0.77 & 0.72 & 0.66 & 0.63  & 0.47 & 0.33  & 0.25  & 0.37 & 0.22 & 0.15 \\
  \hline
  \parbox[t]{2mm}{\multirow{3}{*}{\rotatebox[origin=c]{90}{$t=0.15$}}}& $QP(X_{ILP}^{(t)})$ & 0.99 & 0.99 & 0.99 & 0.99 & 0.99 & 0.99 & 0.98 & 0.98 & 0.98 & 0.98 & 0.97 & 0.97\\
  & $RAER(X_{BP},X_{ILP}^{(t)})$ & 0.96 & 0.95 & 0.95 & 0.96 & 0.95 & 0.92 &  0.88 & 0.75  & 0.66 & 0.86 & 0.68 & 0.57 \\
  & $FR(X_{BP},X_{ILP}^{(t)})$ & 0.79 & 0.77 & 0.77 & 0.74 & 0.69 & 0.65  & 0.55 & 0.38  & 0.30  & 0.53 & 0.32 & 0.22 \\
  \hline
\end{tabular}
\label{tab:downsample}
\end{center}
\end{table*}

{\bf Experiment Setting.} 
%
For every instance in our synthetic dataset, we use different methods to retrieve feasible peer-reviewer matchings.
In particular, we consider:
\begin{enumerate*}[label=(\roman*)]
\item The matching obtained by solving the ILP problem described in Problem \ref{pr:classic}.
We call this solution \textit{Pure-Quality solution} and denote it with $X_{ILP}$.
\item The matching obtained by the greedy heuristic we presented above. We denote this matching with $X_{BP}$.
\item The matching obtained by considering an ILP problem whose objective function tries to maximize the Quality and minimize the Reviewers' Effort at the same time. 
Given $t$, we say that $X$ is a \textit{$t$-tuned solution} if it maximizes $\langle (1-t)W_E+t(-W_R),X \rangle$ under the constraints $l_i \le \sum_{j\in [m]}X_{i,j} \le u_i$ and $\sum_{i\in [n]} X_{i,j}\le U_j$.\footnote{We take $-W_R$ instead of $W_R$ because the LLp is a minimization problem.}
We use $X_{ILP}^{(t)}$ to denote the $t$-tuned solution.
%
%
%
\end{enumerate*}

%
Finally, we fix the parameters we use in our simulations.
%
We set the values $u_i$ and $l_i$ respectively to be $3$ and $5$, $\forall i\in [n]$.
We assume $U_j$ does not depend on $j\in [m]$ and let $U_j$ range in $\{6, 8\}$.
Similarly, we assume $\phi_j$ does not depend on $j\in [m]$ and let its value vary in a set that depends on the value of $U_j$.
We assume the degree of freedom to be proportional to the value $U_j$.
In particular, we consider $\phi_j$ equals to the $50\%$, the $75\%$, and the $100\%$ of $U_j$, and thus, for $U_j=8$, we have $\phi\in\{4, 6, 8\}$.
In the Aligned case, we let the variance of the Gaussian noise, namely $\sigma$, vary in the set $\{0.1,0.3\}$, since we want $X$ to be only a small influence.
Due to the fact that the main goal of the editor is to retrieve high quality solutions and since there does not exist a unique best parameter $t$ to retrieve the $t$-tuned solutions, we let $t$ range in $\{0.05,0.1,0.15\}$.
In that regard, it is worth noticing that the best parameter $t$ might change depending on the dimensions of the problem or on the reviewers' characteristics.
Finally, since all the effort matrices are randomly generated, we run the experiment $250$ times for every set of parameters and report the average of our results.
Due to space limits, we report only the results for the case of $U_j=8$.
All the missing results are deferred in Appendix C.

{\bf Implementation Details.} Our experiments are conducted in the Red Hat Linux Server with Intel Xeon Gold 6138 CPUs.
We use Julia \cite{Julia-2017} as the main coding language and solve the ILP problems by using the JuMP package and HiGHS as an optimization solver \cite{DunningHuchetteLubin2017}.

\subsection{Comparison Metrics}
The metrics we use to evaluate our results are the \textit{Quality Percentage ($QP$)}, the \textit{Reviewers'Average Effort Ratio ($RAER$)}, and the \textit{Fairness Ratio ($FR$)}, which we detail in the following.
{\bf Quality Percentage.}
Given a peer-reviewing allocation $X$ and a quality matrix $W_E$, we define the quality of $X$ as $\QQ(X):=\langle W_E,X \rangle$.
Given a solution $X_{ILP}$ of the ILP problem and a matching $X$ found, we define the $QP$ as $QP(X):=\frac{\QQ(X)}{\QQ(X_{ILP})}$.
Since $X_{ILP}$ does maximize the quality over the set of feasible matchings, we have that $QP(X)\in [0,1]$ for every feasible matching.
%
%
%

%
{\bf Reviewers' Average Effort Ratio.}
Given a peer-reviewing allocation $X$ and an effort matrix $W_R$, we define the effort required to reviewer $j$ from $X$ as $\rho_j(X):=\langle (W_R)_{:,j}, X_{:,j} \rangle$ and say that $j$ is \textit{active} if $\rho_j(X)>0$.
We define the total reviewers' effort required by $X$ as $\EE(X):=\langle W_R,X \rangle=\sum_{j\in [m]}\rho_j(X)$.
Finally, we define the average effort required by the active reviewers as $\EE_{avg}(X):=\frac{\EE(X)}{n_{act}(X)}$, where $n_{act}(X)$ is the number of active reviewers according to $X$.
Given two peer-reviewers allocations $X$ and $X'$, we define the Reviewers' Average Effort Ratio as $RAER(X,X')=\frac{\EE_{avg}(X)}{\EE_{avg}(X')}$.
%
%

%
{\bf Fairness Ratio.} Given an allocation matrix $X$, let $\rho$ be the $n_{act}(X)$-dimensional vector containing the efforts of all the active reviewers.
To measure how fairly spread the effort induced by $X$ is, we use the variance of vector $\rho$, denoted as $\theta(X)$ \cite{jain1984quantitative}.
Given two peer-reviewers allocations $X$ and $Y$, we define the Fair Ratio between $X$ and $X'$ as $FR(X,X')=\frac{\theta(X)}{\theta(X')}$.
%

\subsection{Results}

In this section, we compare the performance of the three matchings described in Section \ref{subsect:exp}.
First, we comment how $X_{BP}$ compares with $X_{ILP}$ and then compare $X_{BP}$ with $X^{(t)}_{ILP}$.
{\bf The Pure-Quality Solutions Against the Greedy Heuristic.}
We first comment on the {Aligned case}.
Notice that, in the Aligned case, maximizing the quality and minimizing the effort are similar tasks, thus this represents the best scenario for the ILP model.
%
%
Nonetheless, we observe that our model is able to find solutions whose $QP$ is consistently larger than $95\%$, such that
\begin{enumerate*}[label=(\roman*)]
\item the $RAER$ is $10\%$ to $20\%$ lower and
\item the total effort is more evenly spread. 
\end{enumerate*}
Indeed, we observe that the $FR$ ranges from $76\%$ to $48\%$, so that the variance obtained by the ILP model is $1.3$ times the variance obtained by our model in the best case and more than double in the worst case.
We, therefore, conclude that the personal effort required from every reviewer according to the solution of the BP model is lower on average, spread amongst more reviewers, and, overall, more fairly distributed.
%
%

We now comment on the {Random case}, \textit{i.e.} the case in which there is no guaranteed relationship between the quality of a matching and the effort of the agents.
%
%
For both the uniform and exponential distribution, we observe that the $QP$ slightly drops between from $94\%$ to $90\%$.
However, the quality drop comes with a larger drop in the $RAER$ and an even larger drop in the $FR$.
Indeed, the allocation found by our heuristic halves the effort required by every reviewer in most cases, and in some instances, it achieves almost a third of the effort required by the ILP solution.
Similarly, the $FR$ drops between $0.35$ and $0.07$ hence the variance of the heuristic solution is up to $15$ times lower than the one found by the classic ILP model.
All the results of our experiments are reported in Table \ref{tab:downsample_tr}. 

%

{\bf The $t$-tuned Solutions Against the Greedy Heuristic.}
Finally, we compare our heuristic to the $t$-tuned solutions for $t = 0.05,0.1,0.15$.
As we observe from the results, taking a convex combination of the effort matrix and the quality matrix leads to solutions that better balance the quality and the effort. 
This is clear when we look at the performance on the aligned dataset: in this case, the $QP$ of both $X_{BP}$ and $X^{(t)}_{ILP}$ are comparable for every $t$.
Similarly, the $RAER$ between $X_{BP}$ and $X_{ILP}^{(t)}$, ranges between $90\%$ and $95\%$, hence the $X_{BP}$ still requires less effort to the reviewers, but only marginally.
The real difference between the two solutions can be appreciated when look at their $FR$.
Indeed, our heuristic is capable of finding a matching that spreads the effort in a more fair way amongst the reviewers.
It is again worth stressing, that the aligned case represents the best set of conditions for any ILP model, since maximizing the quality and minimizing the effort do coincide in this specific framework.
This consideration is confirmed by the results we get for the Random dataset.
In this case, the performances of $X^{(t)}_{ILP}$ is more in line with the performance of $X_{ILP}$: the loss in $RAER$ is no longer negligible and the loss in $FR$ increases.
Finally, we observe that using our model has a further advantage over the $t$-tuned solutions: our method does not require to find a tuning parameter $t$ to describe the interplay between the editor's and the reviewers' objective function.
The experiment results are listed in the Table \ref{tab:downsample}. 
%
%

%

%
{\bf Final Remarks.}
In every setting considered, we observe that the variance of our solution is noticeably lower than the one found by any ILP based model.
This is especially interesting since we are able to retrieve those fairer solutions without altering the objective function nor by determining a tuning parameter to mix the quality and the effort matrix.
Finally, we notice that our accuracy results are acceptable since the percentages we found are well above the $83\%$ obtained by other heuristic approaches, as the one proposed in \cite{DBLP:conf/ijcai/AhmedDF17}. 

\section{Conclusion and Future Works}
\label{sec:conclusion}

In this paper, we introduced an integer BP model for the peer review problem.
To the best of our knowledge, this is the first model that finds a peer-reviewing matching using this formalism.
We showed that the LLp and the ULp are well-defined problems and that solution has a neat burden-free property under mild assumptions.
Finally, we defined a heuristic solution and validated it through numerical tests.
From the results, we observe that the heuristic finds fairer solutions that require a lower average effort from the reviewers and achieves a competitive quality without requiring any parameter tuning.
For the future avenue, a first improvement to the model would be to relax the binary constraints of the allocation and search for a probabilistic allocation rather than a deterministic one.
%
%
Another interesting aspect of the peer review problem is determining the expertise of the reviewers by comparing records from conferences.

\ack{This project is partially supported by a Leverhulme Trust Research Project Grant (2021--2024). Jie Zhang is also supported by the EPSRC grant (EP/W014912/1).}

\bibliography{ecai}

\clearpage

\section*{Appendix}

In this appendix, we report the missing proofs, the additional theoretical results, and the additional experimental results of our paper.

\section*{Appendix A -- Proofs}
\label{appendix:a}

In this appendix, we report all the proofs of the results stated in the main body of the paper.

\begin{proof} (\emph{Lemma \ref{property}})
It follows from the following chain of inequalities

\begin{align*}
    \sum_{i=1}^b w_i x_i &= \sum_{i=1}^a w_i x_i + \sum_{i=a+1}^b w_i x_i\\
    &\geq \sum_{i=1}^a w_i x_i + w_a \sum_{i=a+1}^b  x_i\\
&= \sum_{i=1}^a w_i x_i + w_a  a- w_a \sum_{i=1}^a  x_i \\
&= \sum_{i=1}^a \left( w_i x_i + w_a -   w_a x_i\right)\\
&= \sum_{i=1}^a \left( (w_a - w_i)(1- x_i) +    w_i\right)\geq \sum_{i=1}^a   w_i.
\end{align*}
\end{proof}

\begin{proof}(\emph{Theorem \ref{lemma-equivelance}})
It is easy to see that the problems \eqref{lower-level-problem-primal_0} and \eqref{lower-level-problem-primal} are equivalent to the following Linear Programming problems
\begin{equation} \label{lower-level-problem-primal-equ}
\arraycolsep=1.4pt\def\arraystretch{1.25}
\begin{array}{lrl} 
\forall~j\in[m],~~& Y^*_{:,j}\in\underset{Y_{:,j} }{\rm argmin}&~\langle  W_{:,j}, Y_{:,j}\rangle,\\
&{\rm s.t.}~&~\sum_{i\in [n]}  Y_{i,j} = U_{i},\\
&\quad &0\le Y_{i,j} \leq Z_{i,j},
\end{array} 
\end{equation}
and the following Integer Linear Programming problems
\begin{equation} \label{lower-level-problem-relax-equ}
\arraycolsep=1.4pt\def\arraystretch{1.25}
\begin{array}{lrl} 
\forall~j\in[m],~~& \widehat{Y}_{:,j}\in\underset{Y_{:,j} }{\rm argmin}&~\langle  W_{:,j}, Y_{:,j}\rangle,\\
&{\rm s.t.}~&~\sum_{i\in [n]} Y_{i,j} = U_j,\\
&\quad &Y_{i,j} \leq Z_{i,j},\\
&\quad &Y_{:,j}\in \mathbb{B}_{1,m},
\end{array} 
\end{equation}
respectively. 
Therefore, to show that $\langle  W, Y^*\rangle=\langle  W, \widehat Y\rangle$, we prove $\langle  W_{:,j}, Y^*_{:,j}\rangle=\langle  W_{:,j}, \widehat Y_{:,j}\rangle$ for every $j\in[m]$. Without loss of generality, we only consider the case $j=1$.
Furthermore, we assume $Z_{i,1}=1$ for every $i\in [n]$, this can be done without loss of generality since otherwise, it suffices to restrict the sets of indices.

Since the feasible region of \eqref{lower-level-problem-primal-equ} is a superset of the feasible set of \eqref{lower-level-problem-relax-equ}, it follows $\langle  W_{:,1}, Y^*_{:,1}\rangle\geq\langle  W_{:,1}, \widehat Y_{:,1}\rangle$.
To conclude the proof, we show the other inequality.
If $\widehat{Y}_{:,1}\in  \B_{n,1}$, we conclude the proof, therefore, we assume $\widehat{Y}_{:,1}\notin  \B_{n,1}$. 
Let us now define
\begin{equation}
    \label{S1-S2}
S: = \Big\{i\in[n]~ :~ 0< \widehat Y_{i,1} <1\Big\}
\end{equation}
and
\begin{equation}
    K: = \Big\{i\in[n]~ : ~ \widehat Y_{i,1} \in\{0,1\}\Big\}.
\end{equation}
By assumption, $S\neq \emptyset$ and, by construction, we have
\begin{equation}
    \label{S1S2t}
 U_1=\sum_{i\in [n]}\widehat Y_{i,1}= \underset{i\in S}\sum \widehat Y_{i,1} +   \underset{i\in K}\sum  \widehat Y_{i,1} =: \underset{i\in S}\sum  \widehat Y_{i,1} +t.
\end{equation}
Since $\widehat Y_{i,1} \in\{0,1\}, ~\forall i \in K$, we get that $t$ is an integer number. 
Let us now define $b:=|S|$ and let $a=U_1-t>0$. From equation \eqref{S1S2t} and $ 0< \widehat Y_{1,j} <1$, we infer that $b>a$.
Now, we increasingly reorder the elements in $\{W_{i,1}:i\in S\}$, so that 
\begin{equation} 
\label{cond-0}
{W}_{i_1,1} \leq {W}_{i_2,1} \leq\cdots\leq{W}_{i_a,1}\leq \cdots \leq  {W}_{i_b,1},
\end{equation} 
where $\{i_1,\cdots,i_b\}=S$. Finally, we define $\overline{Y}$ as 
\begin{eqnarray}
\label{YYYYYYY}
\overline Y_{i,1}=\left\{
\arraycolsep=1.4pt\def\arraystretch{1.5}
 \begin{array}{llllll} 
 \widehat Y_{i,1},&~~~~ \text{for}~ i\in K,\\
 1,&~~~~ i\in \{i_1,i_2,\cdots,i_a\},\\
 0,&~~~~ \text{otherwise} 
.\end{array}
\right.
\end{eqnarray}
A simple computation shows that

 \begin{eqnarray} 
\label{cond-1}&& \sum_ {i\in S} \overline Y_{i,1} = \sum_ {i = i_1} ^{i_a} 1 =a,\\ 
\label{cond__2}&&\sum_ {i\in S} {W}_{i,1} \overline Y_{i,1}  =\sum_ {i = i_1} ^{i_a} {W}_{i,1} \\
\nonumber &&\quad\quad\quad\quad\quad\quad\leq \sum_ {i = i_1} ^{i_b}   {W}_{i,1} \widehat Y_{i,1} \leq  \sum_ {j\in S}  {W}_{i,1} \widehat Y_{i,1},
\end{eqnarray} 

which, in conjunction with relationships \eqref{YYYYYYY}, \eqref{cond__2}, and \eqref{S1S2t}, allows us to conclude

\begin{align*}
    \sum_{i\in [n]} \overline{Y}_{i,1}&= \underset{i\in S}\sum \overline{Y}_{i,1} + \underset{i\in K}\sum\overline{Y}_{i,1}\\
    &= \underset{i\in S}\sum  \overline Y_{i,1} + \underset{i\in K}\sum  \widehat Y_{i,1}= a+t= U_1,
\end{align*}

which means that $ \overline Y$ is feasible. Moreover,

 \begin{eqnarray} \label{cond-5}
\langle {W}_{:,1}, \overline Y_{:,1} \rangle&=&   \sum_ {i\in S} {W}_{i,1} \overline Y_{i,1} +  \sum_ {i\in K} {W}_{i,1} \overline Y_{i,1}    \nonumber \\
&\leq& \sum_ {i\in S} {W}_{i,1} \widehat Y_{i,1} +  \sum_ {i\in K} {W}_{i,1} \widehat Y_{i,1} ~~~~~(\text{due to  \eqref{cond__2})}\nonumber\\
 &=& \langle {W}_{:,1}, \widehat Y_{:,1} \rangle,
\end{eqnarray} 
hence $\langle {W}_{1 :}, \overline Y_{1 :} \rangle = \langle {W}_{1 :}, \widehat Y_{1 :} \rangle$.
The uniqueness result follows from the uniqueness of the non-decreasing ordering of the values $\{W_{i,1}:i\in S\}$.
\end{proof}

\begin{proof}(\emph{Theorem \ref{prop:prob_simp}})
First of all, we notice that, since we are considering a minimization problem, imposing the constraint $\sum_{i\in [m]}Y_{i,j}=U_{j}$ for every $j\in [m]$ is equivalent to impose $\sum_{i\in [m]}Y_{i,j}\ge U_{j}$ for every $j\in [m]$.
If $Y$ minimizes the effort of the reviewers component-wise, it also minimizes the total effort of the reviewers.
Toward a contradiction, assume that $Y$ minimizes the total effort of the reviewers, but does not minimize the effort of a reviewer, namely $j$.
From Theorem \ref{lemma-equivelance}, there exists a binary vector $\tilde{Y}_j$ that minimizes the effort of reviewer $j$.
Let us now define the matrix $\bar{Y}$ as it follows: every column different from the $j$-th is equal to the respective column of $Y$, while the $j$-th column is equal to $\tilde{Y}_j$.
It is easy to see that $\bar{Y}$ is still feasible and that the reviewers' effort of $\bar{Y}$ is lower than the effort of $Y$, which contradicts the hypothesis.
\end{proof}

\begin{proof}(\emph{Theorem \ref{prop:feasibility}})
Since the set of feasible triplets for the BP problem is discrete, if we show that it is also not empty, we infer that there exists an optimal solution, which concludes the proof.
Since the ILP problem is feasible, let $X$ be a solution to Problem \ref{pr:classic}.
Since we have that $\max_{j\in [m]}\;\phi_j + 2\max_{j\in [m]} \; U_j\le n$, we can find a $Z$ such that $X\le E-Z$.
Notice that, if the editor proposes $Z$, the matching $X$ is a feasible allocation for the ULp regardless of what the solution to the LLp $Y$ is.
Finally, from Theorem \ref{lemma-equivelance}, we have that there exists an optimal $Y$ that solves the LLp.
Therefore $(Z,Y,X)$ is a feasible triplet for the BP problem, which concludes the proof.
\end{proof}

\begin{proof}(\emph{Proposition \ref{prop:dots}})
Let $(X,Y,Z)$ be a perfect solution, so that $X\le Y\le Z$ holds.
Moreover, let us fix $\Phi:=(\phi_1,\dots,\phi_m)$.
Since $Z$ has at least $U_j+\phi_j$ non-null entries in every row $j\in [m]$, and since $X_{:,j}\le Z_{:,j}$ we infer that the allocation $X$ is $\Phi$-burden-free, which proves the first half of the proposition.
We now focus the second half of the proof.
Let $(X,Y,Z)$ be a perfect solution such that $X$ maximizes the quality of the matching.
Since $X$ maximizes the quality, then $X$ is a solution of the ILP defined in Problem \ref{pr:classic}.
Since we have already proven that $X$ is also $\Phi$-burden-free, this concludes the proof.
To prove the other implication, let $X$ be an optimal and $\Phi$-burden-free solution to Problem \ref{pr:classic}. 
Let $T_j$ be a set that contains $\phi_j$ papers that are worse, according to reviewer $j$, than the ones $j$ is allocated with by assignment $X$.
Since the solution is $\Phi$-burden-free, the set $T_j$ exists for every reviewer $j\in [m]$.
Finally, let us consider a feasible $Z$ such that $Z_{:,j}\ge X_{:,j}+T_j$.
Regardless of the LLp solution $Y$, it holds $X\le Y$, which concludes the second part of the proof.
\end{proof}

\begin{proof}(\emph{Proposition \ref{prop:dict}})
First, we prove that every solution to the classic ILP problem induces at least a solution to the BP problem, and then we show that also the reverse implication holds.
Let us consider $X$ a solution to the classic ILP model in Problem \eqref{pr:classic} and let $Z$ be a binary matrix such that $X\le Z$ and $\sum_{i\in [n]}Z_{i,j}=U_j$ for every $j\in [m]$.
Since every solution to Problem \eqref{pr:classic} satisfies $\sum_{i\in [n]}X_{i,j}\le U_j$ for every $j\in [m]$, such $Z$ always exists.
Let us now define the triplet $(X,Y,Z)$, where $Y=Z$.
Since $X\le Z=Y$, we infer that $\langle X,Y \rangle =\langle X,X \rangle$, hence $AC(X,Y,Z)=1$.
Moreover, since $Y=Z$ is the only feasible matrix of the LLp, it is also optimal for the LLp.
To conclude, notice that $X$ maximizes, by construction, both the functional $\langle W_E, X\rangle$ and $\langle Y,X \rangle$, therefore it is optimal for the BP problem.
Let us now consider a solution to \eqref{problem:bilevel}, namely $(X,Y,Z)$.
If $X$ is not a solution to the ILP problem, we can construct a better solution to the BP problem by using the procedure described above, which contradicts the optimality of $(X,Y,Z)$.
\end{proof}

\begin{proof}(\emph{Theorem \ref{thm:feas_cond}})
    We prove the Theorem by showing that there exists at least one feasible matching for the ULp when we fix $Z=Z_g$ and $Y=Y_g$.
    To do so, we show that there exists a feasible matching $X$ such that $X\le E-Z_g$.
    Given the bipartite graph $G=([n]\times [m], E-Z_g)$, we build the auxiliary graph $G'$ as it follows.
    For every element $i\in [n]$, we generate $l_i$ copies of $i$, we denote this set of papers with $\mathfrak{P}'$.
    Similarly, for every $j\in [m]$, we generate $U_j$ copies of $j$, we denote this set of reviewers with $\mathfrak{R}'$.
    The edge set of the bipartite graph, namely $E'$, is as it follows,: $e':=(i',j')\in E'$ if and only if $e=(i,j)\in E$, where $i'$ is a copy of $i$ and $j'$ is a copy of $j$. 
    To conclude the proof, it suffice to show that there exists a perfect matching of $\mathfrak{P}'$ into $\mathfrak{R}'$.
    First notice that $\mathfrak{P}'$ contains $L=\sum_{i\in [n]}l_i$ elements.
    This follows from the Hall's theorem \cite{kierstead1983effective}, since, by hypothesis and construction, every paper in $\mathfrak{P}'$ is connected to at least $L$ elements of $\mathfrak{R}'$.
\end{proof}

\begin{proof}(\emph{Theorem \ref{them:heur_estimate}})
It is easy to see that $$\langle W_E+ Y_{Z_g},X_{Z_g}\rangle\le\langle W_E+ Y^*,X^*\rangle,$$ where $(X^*,Y^*,Z^*)$ is the optimal solution of the BP problem.
By rearranging the terms of the latter inequality, we have
\[
AC(X_g,Y_g,Z_g)-AC(X^*,Y^*,Z^*)\le \langle W_E, X^*\rangle - \langle W_E, X_g \rangle.
\]
Since $AC(X_g,Y_g,Z_g)=1$ and $AC(X^*,Y^*,Z^*)\le 1$, we infer that the left-hand side of the equation is positive, hence
\[
0\le\langle W_E, X^*\rangle - \langle W_E, X_g \rangle \quad \to \quad \langle W_E, X_g \rangle\le\langle W_E, X^*\rangle.
\]
In particular, the quality obtained by the heuristic final assignment $X_{X_g,Y}$ gives a lower bound on the quality attained by the optimal solution of the BP problem. 
\end{proof}

\section*{Appendix B -- The Secondary Variational Problems}
\label{appendix:b}

In this appendix, we analyze the secondary variational problem (SVP) related to the peer review matching problem.
After studying the properties of the solution when we add a small regularization term, we focus on the diversity functional introduced in \cite{DBLP:conf/ijcai/AhmedDF17}.
In particular, we show how the diversity function they proposed does not directly describe the diversity, but rather an SVP induced by the entropy function.
This allows us to give a theoretical explanation of the large Entropy Gain results presented in \cite{DBLP:conf/ijcai/AhmedDF17}.

\subsection*{Enhancing the solutions Via Second Variational Problem}
Albeit maximizing a linear function has been proven to be a reliable way to find a meaningful peer review matching, determining a matching by only using an ILP model leads towards undesirable features.
First, the solution is often non-unique, which means there is still room for improving the solution.
Second, since the solutions to LP problems lay in the vertexes of a polytope, the maximizers of the objective function lack other appealing properties such as fairness, entropy, and diversity.
For example, the maximum allocation described in Example \ref{ex:fairness} allocates all the best papers to one reviewer and leaves the other reviewer with the worst one.
A classic method to deal with both these issues is to augment the objective function with a (usually convex) function as follows
\begin{equation}
    \label{eq:intro_sect4}
    \langle W,X \rangle + \lambda C(X),
\end{equation}
where $C$ is a function that describes the property we are interested in and $\lambda\in \mathbb{R}$ is a parameter that scales the effect of the function $C$ over the optimization problem.
In this appendix, we show that for every problem \eqref{eq:intro_sect4} there does exist a small enough parameter $\lambda$ for which the maximizers of the function \eqref{eq:intro_sect4} are the matchings that maximizes\footnote{or minimizes, depending on the sign of $\lambda$} the function $C$ over the set of matchings that maximize the quality $\langle W,X \rangle$. 
\begin{definition}
\label{def:model_epsilon_div}
Given a matching problem and a convex and bounded function $C$ over the set of feasible matching, we define the secondary variational problem induced by $C$ as
\begin{equation}
\label{eq:problem}
    \max_{X\in \mathcal{A}} \WWl (X) =\max_{X\in \mathcal{A}} \langle W, X\rangle + \lambda C(X),
\end{equation}
where $\mathcal{A}$ is the set of feasible points $X$ and $\lambda\in\mathbb{R}$. 
\end{definition}

Depending on the value of $\lambda$ the solution $X_\lambda$ of \eqref{eq:problem} changes.
First of all, depending on the sign of $\lambda$, the problem searches for the maximum or the minimum of $C$.
Moreover, as the module of $\lambda$ grows, the solution $X_\lambda$ gets less optimal for the LP problem and approaches the maximum or minimum of $C$.
For the sake of simplicity, in what follows, we assume $C$ to be convex and $\lambda$ to be negative, which is also the most common scenario.

\begin{proposition}
\label{prop:2}
Let $W=(w_{i,j})_{i,j}$ be an edge weight matrix such that $\min_{i,j}w_{i,j}>0$ and let $C$ be a convex and bounded function.
Then, there exists $\epsilon>0$ such that for every $\lambda\in(0,\epsilon)$ the solution to
\begin{equation}
    \label{eq:lambdasmall}
    \max_{X\in\mathcal{A}} \langle W , X \rangle -\lambda C(X)
\end{equation}
is a Maximum Edge-weighted Matching with respect to edge weight matrix $W$ that minimizes $C$.
\end{proposition}

\begin{proof}
Let $\bar{X}\in\AA$ be a minimizer of the LP problem
\[
\max_{X\in \AA}\langle W,X \rangle.
\]
Toward a contradiction, assume that $\forall\lambda>0$ there exists a $X_\lambda$ such that
\[
\langle W,\bar{X} \rangle - \lambda C(\bar{X}) < \langle W,X_\lambda \rangle - \lambda C(X_\lambda).
\]
Let us now consider a sequence, namely $\lambda_n$, such that $\lambda_n>0$ for every $n$ and that converges monotonically to $0$.
Let $X_n:=X_{\lambda_n}$ be the related sequence of solutions.
Since $\AA$ is finite, we have that, up to a sub-sequence, $X_n$ converges to some element in $\AA$, namely $X^*$, moreover, starting from a given $N$, we have $X_{n}=X^*$ for every $n>N$.
By assumption, we have that $\langle W,X^*\rangle<\langle W,\bar{X} \rangle$.
Moreover, we have
\[
\langle W,\bar{X} \rangle -\lambda_n C(\bar{X}) <  \langle W,X^*\rangle -\lambda_n C(X^*)
\]
for every $n>N$.
By manipulating the latter relation, we get
\[
0<\langle W,\bar{X}-X^* \rangle  <   \lambda_n (C(\bar{X})-C(X^*))
\]
which is a contradiction, since the left-hand side is positive and the right-hand side converges to zero since $C$ is bounded.
\end{proof}

If the matrix $W$ has integer values, we can bound the value of $\epsilon$.

\begin{proposition}
\label{prop:bound_on_eps}
If the matrix $W$ has integer entries and $C$ is a bounded and convex function, we have that $\epsilon$ is greater than $\frac{1}{\Delta C}$ where $\Delta C= \max_{\mathcal{A}} \, C - \min_{\mathcal{A}} \, C$ and $\mathcal{A}$ is the set of feasible binary matrices. 
In particular, whenever $0<\lambda<\frac{1}{\Delta C}$, the solutions of \eqref{eq:lambdasmall} are also maximizers of the LP model defined in Problem \ref{pr:classic}.
\end{proposition}

\begin{proof}(\emph{Proposition \ref{prop:bound_on_eps}})
Let us consider an optimal solution $\bar{X}$ for the ILP model.
Since $W$ has integer values, for any given $X\in \AA$ that is not optimal, we have $\langle W,\bar{X}-X \rangle \ge 1$, hence we have that if there exists $X\in \AA$ that solves \eqref{def:model_epsilon_div} and that does not maximize the quantity $\langle W,X \rangle$ over $\AA$, then we have
\[
1 \le \langle W,\bar{X}-X \rangle \le \lambda\Big( C(\bar{X})-C(X) \Big)\le \lambda \Delta C
\]
which is not possible whenever $\Delta C \le \frac{1}{\lambda}$.
\end{proof}

\subsection*{The SVP in the Bilevel Formulation}

Due to the nature of the BP problems, we can add a convex penalizing function to both the ULp and the LLp, which leads us to the following BP problem.

\begin{problem}
\label{problem:bilevel_2}
In the framework described in Section \ref{sect:bilevel}, consider the following problem
\begin{eqnarray}
\label{eq:bivel_prob_pen}
\arraycolsep=2pt\def\arraystretch{1.25}
\begin{array}{ll} 
\underset{Z,X\in \mathbb{B}_{nm}}\max &\;\;\; \langle W_E,X\rangle + \langle Y^*,X \rangle + a_0 C_E(X)\\
&{\rm where }\quad l_i \le\sum_{j\in [m]}X_{i,j}\le u_i, \\
&\quad\quad \sum_{i\in [n]}X_{i,j}\le U_j,\\
&\quad\quad\sum_{i\in [n]}Z_{i,j}=U_j+\phi_j,\quad \text{and}\\
&\quad\quad X \le E-Z+Y^*\\
    &Y^* \in\underset{ Y\in \mathbb{B}_{nm}}\argmin \; \langle W_R,Y\rangle + a_j C_j(Y)\\
    &\quad{\rm where}\quad \sum_{i\in [n]}Y_{i,j}=U_j \quad \text{and}\quad 0\le Y\le Z 
\end{array} 
\end{eqnarray}
where $C_E$ is the penalizer of the ULp, $C_j$ are the penalizers of the LLp, and $\{a_k\}_{k=0,\dots,m}$ are parameters that scale the influence of the respective penalizers. 
As for the previous model, the LLp is a component-wise minimization, that is $Y^*$ is the matrix whose columns minimize the quantity $\langle (W_R)_{:,j},Y_{:,j} \rangle + a_j C_j(Y_{:,j})$.
\end{problem}


\subsection*{Some Remarks on the Diversity Function}

Finally, we study the diversity function.
To the best of our knowledge, diversity was first introduced in the context of peer review matching in \cite{DBLP:conf/ijcai/AhmedDF17}.
Given a bipartite graph $(A\cup B,A\times B)$, a function that evaluates the weights of the edges $W:A\times B\to [0,\infty)$, and a matching $X=\{X_{i,j}\}_{(i,j)\in A\times B}$, the authors define the diversity of $X$ as
\begin{equation}
    \label{eq:diversity_weight_app}
    F(X):=\sum_{i\in A}\sum_{k\in K}\Big(\sum_{j\in B_k}w_{i,j}X_{i,j}\Big)^2
\end{equation}
where $w_{i,j}=W(i,j)$ and $\BB:=\{B_k\}_{k\in K}$ is a partition of one of the sides of the bipartite graph, in this case, $B$.
They then use a greedy heuristic method to minimize $F$ over the set of constraints we described in Section \ref{sec:basic_notions}.
Through empirical experiments, they show that the minimizers of \eqref{eq:diversity_weight_app} are close to the classic optimums from a utility viewpoint.
However, their solution has a way higher level of entropy.\footnote{For the sake of convenience, the authors of \cite{DBLP:conf/ijcai/AhmedDF17} considered the minimization version of the ILP rather than the maximizing one we presented in the paper. In this case, the editor aims at minimizing the quantity $\sum_{i,j}w_{i,j}X_{i,j}$.}
Although at first sight, it might look like this model has nothing to do with the SVP, there exists a deep relation between the problem considered in \cite{DBLP:conf/ijcai/AhmedDF17} and the SVP induced by the following entropy function
\begin{equation}
    \label{eq:entropy}
    C(X):=-\sum_{i\in [n]}\sum_{k\in [K]}\Big(\sum_{j\in B_k} w_{i,j}X_{i,j}\Big)\log \Big(\sum_{j\in B_k} w_{i,j}X_{i,j}\Big).
\end{equation}
Moreover, we also show that the function $F$ defined in \eqref{eq:diversity_weight_app} is not related to the classic notion of diversity (as the next example shows) and propose a correct definition of the diversity function.
%




\begin{figure*}[ht!]
\begin{center}

\subfloat[][The initial complete bipartite graph.]{\includegraphics[width = 0.25\textwidth]{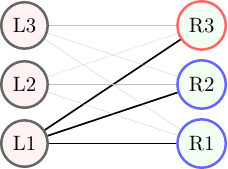}} 
\hfil
\subfloat[The graph describing the proposal $Z$ of the editor.]{\includegraphics[width = 0.25\textwidth]{fefsdsfd.pdf}}
\hfil
\subfloat[The graph describing the bidding of the reviewers $Y$.]{\includegraphics[width = 0.25\textwidth]{fefsdsfd.pdf}}
\end{center}
\caption{The graph described in Example \ref{ex:1}. The partition on the right side of the graph is described by the different colors of the circles.\label{fig:1}
}
\end{figure*}

\begin{example}
\label{ex:1}
Let us consider the bipartite graph defined by $A=\{L1,L2,L3\}$ and $B=\{R1,R2,R3\}$ (see Figure \ref{fig:1}).
Let us set $w_{i,j}=W(Li,Rj)$ where $i,j=1,2,3$ and let us assume that $w_{1,1}=w_{1,2}=0.1$ and $w_{1,3}=0.9$.
Finally, we assume that the partition over $B$ is $B_1:=\{R1,R2\}$ and $B_2=\{R3\}$.
Let us assume that every element of $A$ has to be matched to at least two items of $B$.
It is then easy to see that the optimal allocation (the one that minimizes the functional in \eqref{eq:diversity_weight_app}) of $L1$ is to allocate it to $R1$ and $R2$ since $(0.1+0.1)^2<0.1^2+0.9^2$.
However, this is counterintuitive since the diversity function should seek allocations that match elements of $A$ to as many elements of different clusters in $B$.
It is easy to generalize the example to the case of generic positive weights.
Finally, notice that if all the weights of the graph are the same (say, for example, equal to $1$), the notion of diversity is consistent with what has been proposed in \cite{DBLP:conf/ijcai/AhmedDF17}).
\end{example}
In particular, minimizing the functional \eqref{eq:diversity_weight_app} does not always lead to a diverse matching.
Motivated by the latter example, we propose the following definition of the diversity functional.

\begin{definition}
Given a bipartite graph $([n]\cup [m],[n]\times [m])$ and a partition $\BB:=\{B_k\}_{k\in K}$ over $[m]$, we define the diversity functional $\DD:\mathbb{B}_{nm}\to [0,+\infty)$ as
\begin{equation}
    \label{eq:diversity_eq_right}
    \DD(X)=\sum_{i\in [n]}\sum_{k\in [K]}\Big(\sum_{j\in B_k}X_{i,j}\Big)^2.
\end{equation}
\end{definition}

From now on, we refer to the function $F$ defined in \eqref{eq:diversity_weight_app} as the \textit{weighted diversity function}, while we refer to $\DD$ defined in \eqref{eq:diversity_eq_right} as the \textit{diversity function}.

\subsection*{A Generalized Class of Weighted Diversity Functions and the Relation with the Entropy}

In this section, we study how the minimization of the weighted diversity function $F$ defined in \eqref{eq:diversity_weight_app} is related to an SVP of the linear problem $\min_{X\in \AA} \;\langle W,X \rangle$ induced by the entropy function \eqref{eq:entropy}.
This allows us to justify the results found in \cite{DBLP:conf/ijcai/AhmedDF17}. 
We start the discussion by introducing a continuous family of functions that generalizes the weighted diversity function $F$.

\begin{definition}
Given a weight matrix $W$ and a parameter $p\ge 1$, we define the weighted $p$-diversity functional as
\[
F_p(X):=\sum_{i\in A}\sum_{k\in K}\Big( \sum_{j\in G_k}w_{i,j}X_{i,j} \Big)^p.
\]
\end{definition}

Notice that $F_2$ is equal to the weighted diversity functional $F$, defined in \eqref{eq:diversity_weight_app}. 
From the same argument used to prove Proposition \ref{prop:2}, we infer the following result.

\begin{proposition}
\label{prop:1}
Given any matching $X$, we have that $\lim_{p\to 1}F_p(X)=\langle W,X\rangle$.
Moreover, the function $p\to F_p(X)$ is continuous and, given any feasible set $\mathcal{A}$, there exists a $\Bar{p}\in(1,+\infty)$ such that, for every $1<p<\bar p$, every minimizer of $F_p$ over $\mathcal{A}$ is a minimizer of $X\to F_1(X)$ over $\mathcal{A}$.
\end{proposition}

In particular, we have that the minimizers of $F_p$ for $p\approx 1$ also maximize the quality of the matching.
This is in line with the results reported in \cite{DBLP:conf/ijcai/AhmedDF17}, as we have that, due to the continuity of the function $p\to F_p(X)$, the minimum of $F_2$ has to be not that different from an element of $\tilde{S}_1$.
Let us denote with $\mathcal{S}_p$ the set of minimizers of $F_p$ over $\mathcal{S}$, from Proposition \ref{prop:1}, we infer that the set $\mathcal{S}_p$ converges to a set $\tilde{S}_1$ that is a subset of the minimizers of $F_1$.
We now show that this process is related to an SVP and that the set $\tilde{S}_1$ is the subset of minimizers of $F_1$ that maximizes the entropic penalizing function introduced in \eqref{eq:entropy}.
Indeed, we recall that, given $a>0$, it holds true the following approximation formula 
\[
a^{1+\epsilon}\approx a(1 + \epsilon\log(a)),
\]
where $\epsilon>0$ is a small parameter.
The latter formula allows us to approximate $F_{1+\epsilon}(X)$ as it follows

\begin{align*}
    F_{1+\epsilon}&(X)= \sum_{i\in [n]}\sum_{k\in [K]}\Big(\sum_{j\in B_k}w_{i,j}X_{i,j} \Big)^{1+\epsilon}\\
    &\approx  \sum_{i\in [n]}\sum_{k\in [K]}\sum_{j\in B_k}w_{i,j}X_{i,j}\Big(1+\epsilon \log\big(\sum_{j\in B_k}w_{i,j}X_{i,j}\big) \Big)\\
    &=\langle W, X\rangle \\
    &\quad+\epsilon \sum_{i\in [n]}\sum_{k\in [K]}\Big(\sum_{j\in B_k}w_{i,j}X_{i,j}\Big) \log\Big(\sum_{j\in B_k}w_{i,j}X_{i,j}\Big).
\end{align*}

Therefore, the weighted diversity function becomes an SVP for the entropy of the weights allocated to the partition.
From Proposition \ref{prop:2}, we infer the following result.

\begin{proposition}
Let us consider a quality matrix $W$ such that $w_{i,j}>0$ for every $i\in [n]$ and $j\in [m]$.
Then, for every $p\in [1,\tilde{p}]$, there exists a value $\Bar{p}>1$ for which every solution to the follwing problem
\begin{equation}
    \label{eq:min_p_small}
    \min_{X\in\mathcal{A}} F_p(X), 
\end{equation}
is a minimizer of $F_1$ that maximizes the entropy function \eqref{eq:entropy}. 
\end{proposition}
Once again, this result is in line with the results presented in \cite{DBLP:conf/ijcai/AhmedDF17}, where the authors show that minimizing $F_2$ instead of the classical linear function ($F_1$ according to our notation) leads to solutions that have a higher entropy while still retaining a high overall quality.
Following this idea, by applying the algorithm to the $p$-diversity functional with $p\in (1,2)$, we should find a solution that has a better efficiency from a matching perspective and that does not alter the entropy gain described in \cite{DBLP:conf/ijcai/AhmedDF17} for the case $p=2$.
%


\section*{Appendix C -- Additional Experiment Results}
\label{appendix:c}

In this appendix, we report the missing  experimental results.
%

In Table \ref{tab:downsample_tr_U6a} and \ref{tab:downsample_tr_U6}, we report the results for $U=6$.
We observe that the results are in accordance with the ones found in Section \ref{sec:num_exp}. For the sake of conciseness, we comment only on the comparison between $X^{(t)}_{ILP}$ and $X_{BP}$. In this case, we have that 
\begin{itemize}
    \item in the Aligned case the quality and the fairness obtained by $X_{BP}$ and $X_{ILP}^{(t)}$ are comparable, even though our heuristic finds slightly fairer solutions. The main difference, in this case, lies in how the effort is spread amongst the reviewer: $X_{BP}$ finds solution whose variance is smaller than any $X^{(t)}_{ILP}$.
    \item In both the random cases, our solution finds solutions whose overall effort is much smaller than the one required by the $X^{(t)}_{ILP}$. Again, the solutions found by our heuristic are much more fair than any $t$-tuned solution.
\end{itemize}

\begin{table*}[t]
\begin{center}
\caption{Quantitative results for different values of $\phi$ and differently generated effort matrices. 
Every column represents a different framework and is characterized by the effort matrix, the random variable used to generate the matrices, and the degree of freedom.
For each framework, we report the averages over $250$ instances of the $QP$ (Quality Percentage), $RAER$ (Reviewers' Average Effort Ratio), $FR$ (Fairness Ratio), and the $AC$ (Accordance Percentage).\label{tab:downsample_tr_U6a}}
%
\begin{tabular}{  c | @{}c@{}|@{}c@{} | @{}c@{}|@{}c@{} | @{}c@{}|@{}c@{}|@{}c@{}|@{}c@{}|@{}c@{}|@{}c@{} |@{}c@{} | @{}c@{} }
  \hline
  & \multicolumn{6}{c |}{Alligned } & \multicolumn{6}{ c }{Random } \\
  \cline{2-13}
  U=6 & 
  \multicolumn{3}{c|}{$\sigma=0.1$} & \multicolumn{3}{c |}{$\sigma=0.3$} & \multicolumn{3}{c|}{$X\sim \mathcal{U}[0,1]$}  & \multicolumn{3}{c }{$X\sim Exp(0.5)$} \\
  \cline{2-13}
  & $\;\;\phi=3\;\;$ & 
  $\;\;\phi=4\;\;$ & $\;\;\phi=6\;\;$ & 
  $\;\;\phi=3\;\;$ & $\;\;\phi=4\;\;$ & 
  $\;\;\phi=6\;\;$ & $\;\;\phi=3\;\;$ & 
  $\;\;\phi=4\;\;$ & $\;\;\phi=6\;\;$ & 
  $\;\;\phi=3\;\;$ & $\;\;\phi=4\;\;$ & 
  $\;\;\phi=6\;\;$ \\
  \hline
  $QP(X_{BP})$ & 0.99 & 0.98 & 0.98 & 0.97 & 0.97 & 0.96 & 0.94 & 0.93 & 0.90 & 0.94 & 0.93 & 0.90\\
  $RAER(X_{BP},X_{ILP})$ & 0.93 & 0.93 & 0.92 & 0.88 & 0.87 & 0.85   & 0.66 & 0.59 & 0.49 & 0.48 & 0.42 & 0.32 \\
  $FR(X_{BP},X_{ILP})$ & 0.71 & 0.70 & 0.68 & 0.64 & 0.61 & 0.57 & 0.42 & 0.34 & 0.24 & 0.22 & 0.16 & 0.09 \\
  $AC$ & 0.95 & 0.94 & 0.95 & 0.94 & 0.94 & 0.95 & 0.96 & 0.97 & 0.98 & 0.96 & 0.97 & 0.98\\
  \hline
 
\end{tabular}

\end{center}
\end{table*}

\begin{table*}[ht!]
\begin{center}
\caption[Experiments]{Quantitative results for different values of $\phi$ and differently generated effort matrices. 
Every column represents a different framework and is characterized by the effort matrix, the random variable used to generate the matrices, and the degree of freedom.
%
%
For each framework, we report the averages over $250$ instances of the $QP$ (Quality Percentage), the $RAER$ (Reviewers' Average Effort Ratio), and the $FR$ (Fairness Ratio).
We recall that the quality percentage is computed with respect to the maximum achievable quality, \textit{i.e.} the one achieved by $X_{ILP}$.
The $RAER$ and the $FR$ are instead computed by comparing $X_{BP}$ and $X^{(t)}_{ILP}$. \label{tab:downsample_tr_U6}
}

\begin{tabular}{  @{}c | c | @{}c@{}|@{}c@{} | @{}c@{}|@{}c@{} | @{}c@{}|@{}c@{}|@{}c@{}|@{}c@{}|@{}c@{}|@{}c@{} |@{}c@{} | @{}c@{} }
  \hline
   \multicolumn{2}{c |}{ } &  \multicolumn{6}{c |}{Alligned } & \multicolumn{6}{ c }{Random } \\
  \cline{3-14}
  \multicolumn{2}{c |}{ U=6 } & 
  \multicolumn{3}{c|}{$\sigma=0.1$} & \multicolumn{3}{c |}{$\sigma=0.3$} & \multicolumn{3}{c|}{$X\sim \mathcal{U}[0,1]$}  & \multicolumn{3}{c }{$X\sim Exp(0.5)$} \\
  \cline{3-14}
  \multicolumn{2}{c |}{} & $\;\;\phi=4\;\;$ & 
  $\;\;\phi=6\;\;$ & $\;\;\phi=8\;\;$ & 
  $\;\;\phi=4\;\;$ & $\;\;\phi=6\;\;$ & 
  $\;\;\phi=8\;\;$ & $\;\;\phi=4\;\;$ & 
  $\;\;\phi=6\;\;$ & $\;\;\phi=8\;\;$ & 
  $\;\;\phi=4\;\;$ & $\;\;\phi=6\;\;$ & 
  $\;\;\phi=8\;\;$ \\
  \hline
  \parbox[t]{2mm}{\multirow{3}{*}{\rotatebox[origin=c]{90}{$t=0.05$}}}& $QP(X_{ILP}^{(t)})$ & 0.99 & 0.99 & 0.99 & 0.99 & 0.99 & 0.99 & 0.99 & 0.99 & 0.99 & 0.98 & 0.97 & 0.97\\
  & $RAER(X_{BP},X_{ILP}^{(t)})$ & 0.94 & 0.95 & 0.94 & 0.92 & 0.90 & 0.89 &  0.74 & 0.66  & 0.55 & 0.63 & 0.55 & 0.42 \\
  & $FR(X_{BP},X_{ILP}^{(t)})$ & 0.73 & 0.72 & 0.69 & 0.69 & 0.66 & 0.62  & 0.52 & 0.41  & 0.28  & 0.36 & 0.27 & 0.16 \\
  \hline
  \parbox[t]{2mm}{\multirow{3}{*}{\rotatebox[origin=c]{90}{$t=0.1$}}}& $QP(X_{ILP}^{(t)})$ & 0.99 & 0.99 & 0.99 & 0.99 & 0.99 & 0.99 & 0.99 & 0.99 & 0.99 & 0.97 & 0.97 & 0.97\\
  & $RAER(X_{BP},X_{ILP}^{(t)})$ & 0.95 & 0.95 & 0.94 & 0.93 & 0.92 & 0.90 & 0.82 & 0.73  & 0.61 & 0.77 & 0.66 & 0.51 \\
  & $FR(X_{BP},X_{ILP}^{(t)})$ & 0.73 & 0.72 & 0.69 & 0.74 & 0.70 & 0.65  & 0.58 & 0.47  & 0.32  & 0.52 & 0.40 & 0.23 \\
  \hline
  \parbox[t]{2mm}{\multirow{3}{*}{\rotatebox[origin=c]{90}{$t=0.15$}}}& $QP(X_{ILP}^{(t)})$ & 0.99 & 0.99 & 0.99 & 0.99 & 0.99 & 0.99 & 0.98 & 0.98 & 0.98 & 0.96 & 0.97 & 0.95\\
  & $RAER(X_{BP},X_{ILP}^{(t)})$ & 0.95 & 0.95 & 0.94 & 0.96 & 0.95 & 0.93 & 0.92 & 0.83  & 0.69 & 0.93 & 0.80 & 0.62 \\
  & $FR(X_{BP},X_{ILP}^{(t)})$ & 0.75 & 0.73 & 0.70 & 0.78 & 0.74 & 0.68  & 0.68 & 0.54  & 0.37  & 0.75 & 0.55 & 0.33 \\
  \hline
\end{tabular}

\end{center}
\end{table*}

\end{document}